\documentclass[aps,prl,twocolumn,letter]{revtex4-2}

\usepackage{graphicx}
\usepackage{dcolumn}
\usepackage{bm}

\usepackage[utf8]{inputenc}
\usepackage[english]{babel}
\usepackage[T1]{fontenc}
\usepackage{amsmath,amssymb,amsfonts}
\usepackage{mathtools}
\usepackage{hyperref}
\usepackage{amsthm}
\usepackage{comment}
\usepackage{tikz}
\usepackage{lipsum}
\usepackage{xcolor}
\usepackage{appendix}
\usepackage{orcidlink}
\usepackage{cleveref}
\crefname{figure}{Fig.}{Figs.}
\crefname{equation}{Eq.}{Eqs.}
\usepackage{braket}

\newtheorem{theorem}{Theorem}
\newtheorem{proposition}{Proposition}

\newtheorem{corollary}{Corollary}

\newcommand{\vct}[1]{\bm{#1}} 
\newcommand{\mtx}[1]{\bm{#1}}
\newcommand{\transp}{\mathrm{T}}

\makeatletter
\newtheorem*{rep@theorem}{\rep@title}
\newcommand{\newreptheorem}[2]{%
\newenvironment{rep#1}[1]{%
 \def\rep@title{#2 \ref{##1}}%
 \begin{rep@theorem}}%
 {\end{rep@theorem}}}
\makeatother

\newreptheorem{theorem}{Theorem}
\newreptheorem{corollary}{Corollary}
\newreptheorem{proposition}{Proposition}

\begin{document}

\title{
A Rigorous Quantum Framework for Inequality-Constrained and Multi-Objective Binary Optimization}

\author{Sebastian Egginger\,\orcidlink{0009-0003-3831-2253}$^{1}$}
\author{Alexander Mandl\,\orcidlink{0000-0003-4502-6119}$^{1,2}$}
\email{alexander.mandl@jku.at}
\author{Kristina Kirova\,\orcidlink{0000-0002-9854-0691}$^{1}$}
\author{Sonja Bruckner\,\orcidlink{0009-0004-8636-4805}$^{2}$}
\author{Stefan Hillmich\,\orcidlink{0000-0003-1089-3263}$^{2}$}
\author{Richard Kueng\,\orcidlink{0000-0002-8291-648X}$^{1}$}
\affiliation{$^1$Department of Quantum Information and Computation at Kepler (QUICK), Johannes Kepler University Linz, Linz, Austria}
\affiliation{$^2$Software Competence Center Hagenberg GmbH, Hagenberg, Austria}

\date{November 2026}

\begin{abstract} 
Encoding objectives into Hamiltonians forms the foundation of quantum optimization. Including inequality constraints into these encodings remains a major challenge. In this work, we show that including inequality constraints is equivalent to having multiple objectives. This insight motivates the Multi-Objective Quantum Approximation (\emph{MOQA}) framework, which comes with rigorous performance guarantees. \emph{MOQA} operates directly at the Hamiltonian level and is compatible with ground state solvers such as quantum annealing or the Quantum Approximate Optimization Algorithm (QAOA).
\end{abstract} 
\maketitle

Binary optimization is regarded as one of the most promising areas for quantum algorithmic advantages that have the potential to matter in practice~\cite{Abbas_2024, Lanes_2025, Koch_2025, Kotil_2025}. Indeed, many important and hard problems in finance~\cite{Orus_2019, Grant_2021, Buonaiuto_2023}, logistics~\cite{Schworm_2024, Yu_2000, Salehi_2022, Azad_2023}, energy management~\cite{Bruckner_2024}, etc.\ can be recast as the task of minimizing an objective function $h(\vct{b})$ (also known as loss or cost function) over an $n$-dimensional binary vector $\vct{b} \in \left\{0,1\right\}^n$.
Most prevalent quantum strategies encode $\vct{b}$ into the computational basis of an $n$-qubit system and the objective into a classical Hamiltonian $\hat{H}$, i.e. ${h(\vct{b})|\vct{b} \rangle=\hat{H} |\vct{b} \rangle}$. Unconstrained minimization then boils down to ground state preparation, which may be done with adiabatic quantum optimization~\cite{Farhi_2000, Farhi_2001, Albash_2018}. The celebrated Quantum Approximate Optimization Algorithm (QAOA)~\cite{Farhi_2014} can be viewed as a Trotterized~\cite{Trotter_1959, Nielsen_2010}, NISQ-friendly~\cite{Preskill_2018} version thereof~\cite{Farhi_2019, Zhou_2020, Blekos_2024}. 

These quantum strategies for binary optimization are typically applied to unconstrained optimization problems with polynomial degree-$k$ objective functions that can readily be recast as a $k$-local classical Hamiltonian on $n$ qubits~\cite{Lucas_2014, Glover_2022, Jun_2023}. Quadratic Unconstrained Binary Optimization (QUBO) problems ($k=2$), and MAX-CUT in particular, fall into this category~\cite{Wang_2018, Farhi_2025}. Most binary optimization problems that arise from applications do, however, feature additional constraints. These either come in the form of equality constraints $f (\vct{b})=0$, or inequality constraints $g (\vct{b}) \geq 0$~\cite{Lucas_2014, Karp_1972}.

Equality constraints $f(\vct{b})=0$ can, at least to some extent, still be handled by regularizing the original objective function $h(\vct{b})$ with a penalty term for constraint violation: $h(\vct{b}) \mapsto h (\vct{b}) + \gamma (f(\vct{b}))^2$~\cite{Lucas_2014, Alessandroni_2025}. 
The optimal solution of this new (unconstrained) objective function is equal to the optimal solution of the equality-constrained original problem, provided that the regularization parameter $\gamma$ is suitably chosen. 
Additional classical optimization subroutines can sometimes handle inequalities, see e.g.~\cite{Diez-Valle_2023, Djidjev_2023, Takabayashi_2025}. Other proposals build on quantum projections~\cite{Herman_2023, Bucher_2025_If}, or use heuristic penalty functions~\cite{Montanez-Barrera_2024, Kanatbekova_2025, Lee_2025, Bottarelli_2025}. 
The most common method to handle inequality constraints in quantum optimization is to use slack variables to reformulate each inequality constraint as an equality constraint, which we know how to handle~\cite{Lucas_2014, Glover_2022, Yarkoni_2022, Sharma_2025}. Many additional qubits are used to represent a binary expansion of each slack variable, so the number of additional qubits grows linearly in the number of constraints and bits of the expansion.

Here, we present a complementary approach that handles inequality constraints without requiring any additional qubits, classical optimization subroutines or additional optimization variables. What is more, our approach is not restricted to specific tasks or solvers and can be equipped with rigorous theoretical underpinning.
Our key idea is to reformulate inequality-constrained minimization of $h(\vct{b})$ subject to $g(\vct{b}) \geq 0$ as an unconstrained minimization of a maximum of two benign cost functions $h_{\max}(\vct{b})={\max\left\{h (\vct{b}),h (\vct{b})-\gamma g(\vct{b})\right\}}$.
This reformulation is achieved by first recasting the inequality as an equality with a non-analytic penalty function: ${g(\vct{b}) \geq 0 \Leftrightarrow \max\left\{ 0,-g(\vct{b})\right\}=0}$ (famously known as ReLU function). Now we have an equality constraint and can regularize the objective function in standard fashion with a suitably chosen $\gamma >0$. This produces $h (\vct{b}) + \gamma \max\left\{ 0,-g(\vct{b})\right\}$ and pulling out the maximum leads to the form advertised above.
This approach straightforwardly generalizes to $M \geq 2$ inequality constraints and produces a maximum of $M+1$ cost functions. So, inequality-constrained binary optimization is really just a special case of unconstrained optimization over the maximum of multiple objectives.

\begin{figure*}[t]
\includegraphics[width=\linewidth]{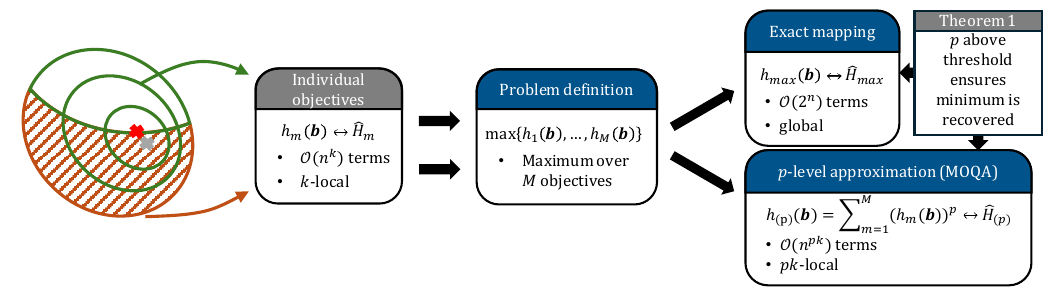}
        \caption{\emph{Concept of the \textit{MOQA} framework.} We start with a problem described by multiple individual objectives. Evaluating their maximum can be done either exactly using exponential resources or approximately with only a polynomial overhead using \emph{MOQA}. \Cref{thm:theorem_1} relates these two approaches and shows that an approximation level $p$ that is independent of the number of optimization variables $n$ suffices to recover the global minimum.}
        \label{fig:concept}
\end{figure*}

The main body of this manuscript starts from there and introduces Multi-Objective Quantum Approximation, or \emph{MOQA} for short (\Cref{fig:concept}).
Motivated by $\ell_p$-norm approximation guarantees for the $\ell_\infty$-norm, the key idea is to approximate the maximization above by the sum of  $p$-th powers of its arguments.
We present how this approximation can be applied to inequality-constrained quantum optimization.
The primary overhead of \emph{MOQA} manifests itself in the locality of individual terms and their increased number, while the number of decision variables stays the same.
Lastly, we underpin our results using numerical evaluations of \emph{MOQA}'s approximation error for randomly generated constrained optimization problems and for an example application solving satisfiability problems.

\section{Results} 

We consider multi-objective binary optimization problems of the form
\begin{align}\label{eq:hmax_def}
\underset{\vct{b} \in \left\{ 0, 1\right\}^n}{\text{minimize}} \quad
h_{\max}(\vct{b}):=\max \left\{ h_1 (\vct{b}),\ldots,h_M(\vct{b})\right\}.
\end{align}
As pointed out above, inequality-constrained binary optimzation can be re-cast in this fashion. Other important and challenging binary optimization problems can also be recast in such a form~\cite{Pan_1997, Lucas_2014}.
By introducing an additive shift to ensure non-negativity of the individual functions ($h_m(\vct{b})\geq 0$ for all $\vct{b}$), their maximum $h_{\max}(\vct{b})$ can be interpreted as the maximum norm $\| \vct{v}(\vct{b})\|_{\ell_\infty}$ of the $M$-dimensional vector $\vct{v}(\vct{b})=\left(h_1 (\vct{b}),\ldots,h_M(\vct{b})\right)^\transp$ subsuming the individual function values.
This maximum norm can be approximated by $\ell_p$-norms and this approximation becomes exact in the limit $p \to \infty$. This suggests the following polynomial approximation ansatz that depends on $p \in \mathbb{N}$:
\begin{align}\label{eq:moqa_obj}
    h_{(p)}(\vct{b}) &\coloneqq \| \vct{v} (\vct{b})\|_{\ell_p}^p = \sum_{m=1}^M h_m^p (\vct{b}).
\end{align}
The required approximation strength, as expressed by the minimal exponent $p$ required to recover the optimum of $h_{\max}$ via minimization of the approximated function $h_{(p)}$, exhibits a threshold behavior.
Remarkably, this threshold is independent of the problem size.
It only grows logarithmically with respect to two metrics: the number of individual functions $M$ and a well-established summary parameter for problem hardness.
\begin{theorem}\label{thm:theorem_1}
    Let $h_{\max}(\vct{b})$ be an objective function defined as the maximum of $M$ objectives $h_m(\vct{b)}$ with minimal argument $\vct{b^*}$ and minimal value $\lambda_1:=h_{\max}(\vct{b^*})$.
    Choosing an approximation level 
    \begin{align}\label{eq:required_p}
        p > p_0 = \frac{\log(M)}{\log(r(h_{\max}) + 1)}
    \end{align}
    ensures that binary vector $\vct{b^*_{(p)}}$ that minimizes $h_{(p)}$ also minimizes $h_{\max}$.
    The parameter $r(h_{\max})$ is the \emph{spectral gap ratio} $(\lambda_2 - \lambda_1) / \lambda_1$, where $\lambda_2$ is the objective function value of the first non-optimal solution.
\end{theorem}
This theorem follows by carefully the controlling approximation accuracy of $h_{(p)}$ based on the objective function values in $h_{\max}$, see Methods for details.

Note that the occurrence of a spectral gap parameter is a natural proxy for hardness in quantum optimization~\cite{Farhi_2000, Albash_2018}. 
Based on this parameter, a concrete threshold is obtained where this approximation recovers the correct optimum.
The spectral gap ratio is in general not known exactly before solving the optimization problem at hand.
Further, the obtained threshold $p_0$ might be too high for a practical implementation of $h_{(p)}$. 
However, exact recovery of the optimum is not always required.
Crucially, we find that for a fixed number of objectives $M$ and a fixed target accuracy for the approximation, the required $p$ once more does not carry any explicit dependence on the problem size.
The following result formalizes this fact in terms of the desired \emph{approximation ratio} $c=h_{\max}(\vct{b})/h_{\max}(\vct{b^*})$.

\begin{corollary}\label{corr:approx}
Let $\vct{b^*}$ be the binary vector that minimizes $h_{\max}(\vct{b})$, an objective function defined as the maximum of $M$ objectives $h_m(\vct{b)}$.
For a desired approximation ratio $c \geq 1$, choosing an approximation level $p > \frac{\log(M)}{\log(c)}$ ensures that the vector $\vct{b^*_{(p)}}$ minimizing $h_{(p)}$ obtains an approximation ratio of at most $c$.
\end{corollary}

\subsection{Implementation on a Quantum Computer}

The most commonly studied class of problems in quantum optimization arises when $h(\vct{b})$ is a quadratic function, i.e.,\ 
\begin{equation}
h(\vct{b})=\vct{b}^\transp \mtx{M} \vct{b},
\label{eq:quadratic-function}
\end{equation}
with the symmetric matrix $\mtx{M} \in \mathbb{R}^{n \times n}$. This class of problems is called Quadratic Unconstrained Binary Optimization (QUBO) and frequently occurs in logistics~\cite{Dantzig_1959, Gonzalez-Bermejo_2022, Chiew_2024}, graph theory~\cite{Farhi_2025, Lucas_2014}, resource allocation~\cite{Kochenberger_2014, Luckow_2021, Bruckner_2024}, and finance~\cite{Markowitz_1952, Glover_2022, Desantis_2024} tasks. 
Such QUBOs are treated in quantum optimization by encoding $h(\vct{b})$ as eigenvalues of a $2$-local Hamiltonian $\hat{H}$: $\hat{H}\!\ket{\vct{b}} = h(\vct{b})\!\ket{\vct{b}}$.
Existing quantum ground state finding algorithms~\cite{Farhi_2014, Albash_2018, Hauke_2020} then provide a way of finding the state of minimal eigenvalue of $\hat{H}$, which corresponds to a minimizer of $h(\vct{b})$.

So far, we have discussed that for the problem of optimization over a maximum (\cref{eq:hmax_def}), under the right conditions \emph{MOQA} returns an approximate cost operator $h_{(p)}$ with its minimum at an assignment $\vct{b^*}$ where $h_{\max}(\vct{b^*})$ is minimal.
In the following, we apply this construction (\cref{eq:moqa_obj}) to quantum optimization.
We assume that all objective functions $h_m(\vct{b})$ in \cref{eq:hmax_def} can be recast as QUBO Hamiltonians $\hat{H}_m$ on $n$ qubits.
Using only a polynomial number of resources in $n$, the operator $\hat{H}_{(p)}=\sum_{m=1}^M \hat{H}_m^p$ then encodes the approximated objective $h_{(p)}$, contrasting the exponential resources required by the exact implementation of $\hat{H}_{\max}$.
It can therefore be used as an efficient approximation of $h_{\max}$ for quantum optimization algorithms.
The following result formalizes this fact. 

\begin{proposition}\label{prop:runtime}
Given a list of $M$ degree-$k$ Hamiltonians in $n$ variables and an approximation level $p$, the degree-$kp$ polynomial $\hat{H}_{(p)}$ approximating $\hat{H}_{\max}$ is computed in time $\mathcal{O}\left(Mn^{kp+1}p\right)$ and contains $\mathcal{O}\left(n^{kp}\right)$ terms. 
\end{proposition}
For the special case of QUBO objectives (${k=2}$), this simplifies to $\mathcal{O}\left(Mn^{2p+1}p\right)$ pre-processing time and $\mathcal{O}(n^{2p})$ terms.
To derive these runtime bounds, an algorithm that constructs $\hat{H}_{(p)}$ as a linear combination of Pauli strings comprised solely of $Z$ and the identity is presented in Methods.
Using existing frameworks~\cite{Li_2022}, this representation can readily be used to implement, for example, evolutions of the form $\exp(-it\hat{H}_{(p)})$ that are required for common quantum optimization algorithms~\cite{Farhi_2014}.

The classical runtime as well as the number of terms remain subexponential as long as $M$ and $p$ only grow polylogarithmically with the problem size $n$. 
If both are actually constant (independent of $n$), then the runtime is polynomial in the problem size $n$. 
Furthermore, if solely recovering solutions of a fixed approximation ratio $c$ is desired, the aforementioned threshold ${p> \log (M)/\log(c)}$ guarantees that $p$ can be chosen constant in $n$ and logarithmic in $M$.

We note, however, that implementing a Hamiltonian consisting of (up to) $2p$-body terms can be challenging on most current quantum hardware. 
Since the vast majority of devices natively support only single- and two-qubit gates, decomposing global interactions into elementary gates significantly increases circuit depth and may involve auxiliary systems. 
In addition, combining $\hat{H}_{(p)}$ with approaches to find its ground state, such as QAOA, which require many layers of such gates, will lead to further overhead. 
Nevertheless, several platforms have proposed or demonstrated multi-qubit entangling gates, including trapped ions~\cite{Lu_2019, Cohen_2015, Katz_2023}, neutral atoms~\cite{Levine_2019, Mohan_2025, Kazemi_2025}, and superconducting circuits~\cite{Roy_2020}. 
Leveraging such hardware-specific capabilities may therefore enable early applications of \emph{MOQA} well before fully fault-tolerant quantum computers are available.

\subsection{Inequality-Constrained Problems}

Having discussed the approximation of generic cost functions $h_{\max}$, we can immediately apply this process to the motivating example outlined in the introduction.
To implement constrained optimization of the function $h(\vct{b})$ subject to the constraint $g(\vct{b}) \geq 0$, using unconstrained solvers implemented by common quantum optimization algorithms, the cost operator 
\begin{align}\label{eq:sing_constraint}
    h_{\max}(\vct{b}) = \max \{ h(\vct{b}), h(\vct{b}) - \gamma g(\vct{b})\}
\end{align}
is approximated using the process described in the previous section.
For an appropriately chosen penalty term $\gamma$, the minimum of the function above coincides with the minimal feasible solution to the posed optimization problem.
This penalty term sets the $\vct{b}$ that violate a constraint further apart from those that do not, but at the same time decreases the (relative) cost function differences among the feasible solutions. 
Consequently, this parameter needs to be chosen with care~\cite{Verma_2022, Alessandroni_2025}.

\begin{figure}[t]
\includegraphics[width=\linewidth]{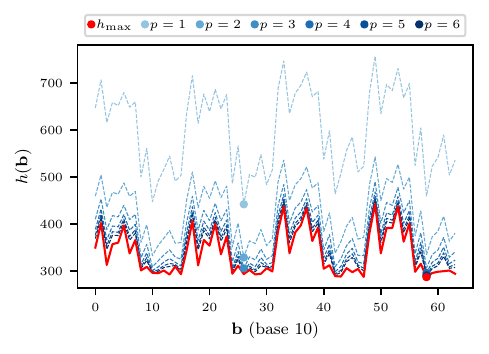}
\vspace{-8mm}
        \caption{To visualize how the approximate landscape approaches the ideal one, a QUBO in $n=6$ variables with a single linear inequality constraint ($\gamma=120$) is displayed. The resulting $h_{\max}(\vct{b})$ is shown as the solid red line.
  This function is consequently approximated via $h_{(p)}(\vct{b})$ displayed by their $p$-th root as blue dashed lines. They become more tightly bound to $h_{\max}(\vct{b})$ as $p$ increases, leading to alignment of their global minima, as shown by the dots.}
        \label{fig:ineq_constraints_a}
\end{figure}

\Cref{fig:ineq_constraints_a} analyzes the treatment of inequality-constrained QUBOs in $n=6$ variables numerically.
There, we use a single linear inequality constraint ${\vct{a}^\transp \vct{b} \geq 0}$. 
The regularized version of this problem can be recast as a maximum of $M=2$ QUBOs, and generic means that we sample all vector and matrix entries independently from a standard Gaussian distribution. 
The thick red line in this figure plots the value of the cost function with $\gamma=120$ against all $2^6=64$ different configurations. 
Different shades of blue then plot tractable \emph{MOQA} approximations to this objective function. We see that the quality of the approximation increases with $p$.

Moreover, \Cref{fig:ineq_constraints_a} also suggests that core features of the overall optimization landscape may already establish themselves for small to intermediate values of $p$. This suggests that choosing an approximation level $p$ below the theoretical threshold from \Cref{thm:theorem_1} may still help us to solve the problem at hand.
In fact, our empirical studies show that \emph{MOQA} still produces reliable results even when we choose a $p$ below the threshold or don't have any known guarantees to begin with. 
To support this claim, we sample $N_s=10000$ random instances of the problem with varying sizes and compare the minima of $h_{\max}$ and $h_{(p)}$ up to $p=20$. In most of these cases, \Cref{thm:theorem_1} would require $p>20$.

\Cref{fig:ineq_constraints_b} evaluates the reliability of the approach using two metrics, which are discussed in the following. The average \emph{success probability} (first plot) measures the probability that an optimal assignment $\vct{b}^*_{(p)}$ for $h_{(p)}$ coincides with a minimum of $h_{\max}$.
Although this is guaranteed only for $p > p_0$ according to \Cref{thm:theorem_1}, this plot shows that even for low $p$, the average success probability continuously increases. 
Furthermore, depending on the application, recovering the exact optimum is not always required.
Instead, obtaining an approximation $h_{(p)}$ whose optimum $\vct{b}^*_{(p)}$ obtains a low cost function value $h_{\max}(\vct{b}^*_{(p)})$ is sufficient.
Consequently, we further evaluate the average \emph{approximation ratio} $h_{\max}(\vct{b}^*_{(p)})/h_{\max}(\vct{b^*}_{\max})$ across all sampled instances (second plot).
We observe that already for small $p$ the approximation is less than $1\%$ off the ideal value.

\begin{figure}[t]
    \includegraphics[width=\linewidth]{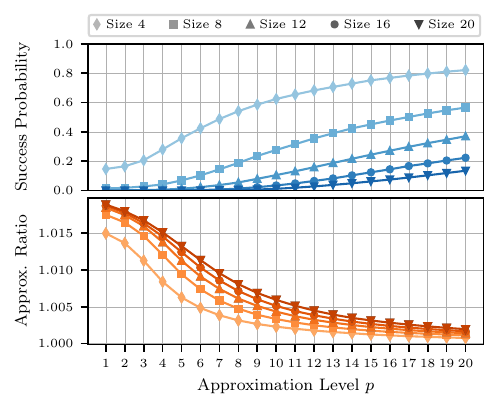}
\vspace{-8mm}
    \caption{Before reaching the threshold in \Cref{thm:theorem_1}, \emph{MOQA} may serve as a heuristic with small $p$.
    This figure shows numerical results based on sampling 10000 random instances of such problems. 
    This allows for an analysis of the success probability in recovering the minimum of $h_{\max}$ (top) as well as of the approximation ratio of the minimum of $h_{(p)}$ (bottom).}
    \label{fig:ineq_constraints_b}
\end{figure}

The analysis above averages over various spectral gap ratios, most of which are below the threshold of \Cref{thm:theorem_1}. To numerically analyze how the error behaves as the approximation level approaches the threshold, we present \Cref{fig:thm1}. 
There, we bring the spectral gap ratios closer to the threshold by reducing $\gamma$ from $120$ to $6$. Then, we examine the success probability and observe that it decreases steadily until the optimum is eventually recovered in all cases before the threshold.
Typically, small $p$, large $n$, and large $M$ increase the deviation between the empirically observed point of maximal success probability and the analytically guaranteed one.

\subsection{Application: Satisfiability}

\begin{figure}
  \begin{centering}
    \includegraphics[width=\linewidth]{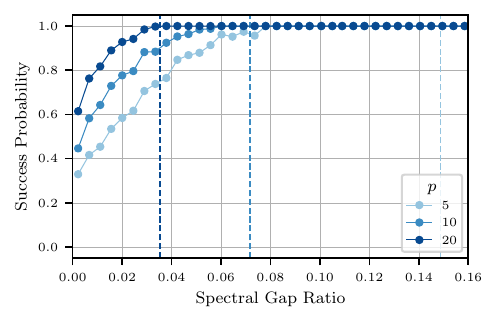}
  \end{centering}
  \vspace{-8mm}
  \caption{The average success probability for recovering optimal solutions of $h_{\max}$ via the approximations $h_{(p)}$. 
  A total of $N_s=10000$ random QUBOs in $n=16$ variables combined with a linear inequality constraint ($\gamma=6$) are analyzed. 
  They are collected into bins based on the resulting spectral gap ratio $r (h_{\max})$. 
  For the three approximation levels $p\in\{5,10,20\}$, the success probability is shown per bin. 
  In addition, the corresponding vertical lines indicate the threshold, after which \Cref{thm:theorem_1} guarantees the recovery of the optimum.}
  \label{fig:thm1}
\end{figure}

By design, \emph{MOQA} also enables us to directly tackle binary optimization problems with multiple constraints. 
Motivated by the large number of optimization problems that can naturally be formalized as multi-constraint optimization~\cite{Werner_2026}, we evaluate this class of problems in the following on the example of the \textsc{Not-All-Equal 3SAT} problem (NAE-3SAT)~\cite{Coja_2012}.
There, the task is to find a satisfying assignment of a Boolean formula in conjunctive normal form containing three literals in each of its $M$ clauses, with the restriction that the variable assignment should not satisfy all literals in a clause at once.
Furthermore, we aim to minimize the number of true variables ($b_i = 1$), leading to the following formalization
\begin{align*}
    \underset{\vct{b} \in \left\{0,1\right\}^n}{\text{minimize}} \quad & h(\vct{b}) := \sum_{i=1}^n b_i\\
    \text{subject to} \quad & 1 \leq \mathrm{Numsat}_m(\vct{b}) \leq 2 \quad \text{for $1\leq m \leq M$},
\end{align*}
where $\mathrm{Numsat}_m(\vct{b})$ counts the number of true literals in clause $m$. 
The lower bound constraint enforces that the clause is satisfied, while the upper bound enforces that not all literals are equal.
As elaborated in the methods section, solving the problem above amounts to solving a linear binary optimization problem subject to $2M$ linear constraints.
To achieve this task, each of the individual constraints is included in the maximization
\begin{align*}
    h_{\max}(\vct{b}) &=\max\{h(\vct{b}),h(\vct{b})-\gamma g_1(\vct{b}),\dots,h(\vct{b})-\gamma g_{2M}(\vct{b})\}.
\end{align*}
Herein, the subtraction of the scaled penalty term ensures that the maximum reduces to $h(\vct{b})$ if and only if all constraints are satisfied.

\Cref{fig:nae_approx_error} plots the accuracy of the approximations $h_{(p)}$ when applied to randomly generated NAE-3SAT instances containing $n=20$ variables.
The instances are obtained by uniformly sampling literals for each clause and post-selecting for satisfiable formulas.
The plot shows the success probability, i.e.~the proportion of instances where the minimal satisfying assignment was recovered via the minimum of $h_{(p)}$, over the ratio $r=M/n$ of the number of clauses to the number of variables in the formula.
As expected, the approximation error decreases with increased approximation level $p$.

We choose to evaluate formulas with $r \in [0.5, 4]$, including instances with $r \approx 2.1$, which have been proven to be the hardest instances of random NAE-3SAT for the case where $n \to \infty$~\cite{Achlioptas_2001,Coja_2012}.
In this critical region, our results highlight that even the moderate approximation level of $p=8$ is sufficient to recover optimal solutions for $n=20$.
Combined with \Cref{prop:runtime}, the degree $kp=8$ polynomial implementing the required cost operator therefore contains $\mathcal{O}(n^8)$ terms even for the hardest problem instance.
Furthermore, the coefficients of $\hat{H}_{(p)}$ can be precomputed classically in polynomial time in $n$. 

\begin{figure}
    \includegraphics[width=\linewidth]{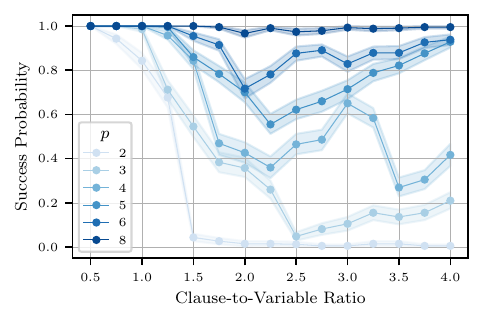}
  \vspace{-8mm}
    \caption{The average success probability for the optimization version of the NAE-3SAT problem for random instances evenly distributed to clause-to-variable ratios between $0.5$ and $4$, with $420$ evaluated instances per point. These objectives contain up to $n=20$ variables and up to $80$ clauses, resulting in up to $160$ constraints for the largest ones.}
    \label{fig:nae_approx_error}
\end{figure}

Overall, these findings underpin MOQA's potential as a valuable extension to the available tools for encoding highly constrained optimization problems on quantum computers.
Lastly, while the selected instances for these results are all satisfiable, we further highlight in the supplementary material that the operators implemented using \emph{MOQA} also reliably distinguish satisfiable instances from unsatisfiable ones.

\section{Discussion}

We consider the problem of mapping a multi-objective objective function $h_{\max} (\vct{b})=\max \left\{ h_1 (\vct{b}),\ldots,h_M (\vct{b}) \right\}$ on $n$ bits $\vct{b} \in \left\{0,1\right\}^n$ to a diagonal $n$-qubit Hamiltonian. 
In many instances, reformulations of problems are necessary to arrive at forms that can be handled using unconstrained optimizers. 
These conversions often come at a price, such as auxiliary qubits for slack variables. 
Due to the structure of MOQA, certain problems directly fit into this framework and don't require reformulation.

We provide both theoretical and empirical arguments that a sum of $p$-th powers $\hat{H}_{(p)}=\sum_{m=1}^M \hat{H}_m^p$ achieves the task of effectively implementing $h_{\max}$, and its approximation accuracy improves with increasing $p$.
If the original objective functions can be realized by $k$-local Hamiltonians $\hat{H}_m$, then $\hat{H}_m^p$ is $pk$-local and contains at most (order) $n^{kp}$ terms. 
Note that randomized compiling (QDRIFT)~\cite{Campbell_2019, Chen_2021} is well-suited to handle such a large number of terms, while recent progress~\cite{Dreier_2025} alleviates the implementation cost of gates with intermediate locality.
Our main technical result states that there is a threshold behavior: choosing $p$ logarithmically in the number of objectives $M$ and inverse logarithmically in the original spectral gap ratio ensures exact recovery of the optimum.

We complement these theoretical findings with empirical studies for both randomly generated QUBOs as well as for random instances of a satisfiability problem. 
These show that for relatively large instances (up to $160$ constraints), \emph{MOQA} only requires a moderate approximation level of $p=8$, while still obtaining a high success probability ($\geq 96\%$) in recovering the optimum.

We emphasize that this (min-max) notion of multi-objective optimization is different from (and narrower than) general multi-objective Pareto front optimization~\cite{Miettinen_1998, Kotil_2025}. Finding these fronts may rely on solving min-max tasks as a subroutine~\cite{Marler_2004}, which establishes a natural bridge between different notions of multi-objective optimization.
Lastly, we intentionally provide only the Hamiltonians $\hat{H}_{(p)}$ and do not (yet) combine them with existing methods to find their ground states, such as QAOA. We do this because we consider the entire pipeline from problem to solution as a modular approach, with \emph{MOQA} as one essential part that enables a plethora of quantum meta-algorithms to extend to a broader range of applications.

\section{Methods}

\subsection{Framework and Theoretical Underpinning}

We approximate the maximum of multiple objective functions in \cref{eq:hmax_def}, restated here for convenience:
\begin{align*}
\underset{\vct{b} \in \left\{ 0, 1\right\}^n}{\text{minimize}} \quad
h_{\max}(\vct{b})=\max \left\{ h_1 (\vct{b}),\ldots,h_M(\vct{b})\right\}.
\end{align*}
We also require that each objective $h_m$ only assumes non-negative values, i.e.\ $h_m (\vct{b}) \geq 0$ for all $\vct{b} \in \left\{ 0,1\right\}^n$.
Their non-negativity allows us to raise $h_{\max}(\vct{b})$ to an integer power $p$ without changing the relative ordering of function values: $h_m (\vct{b}) \leq h_m (\vct{c})$ iff $h_m^p (\vct{b}) \leq h_m^p(\vct{c})$. 
We use this to approximate the pointwise maximum with a degree-$p$ polynomial:
\begin{align}
h_{\max}(\vct{b})^p &= \max \left\{ h_1^p (\vct{b}),\ldots,h_M^p(\vct{b})\right\} \approx h_{(p)} (\vct{b})\label{eq:approximation}\\
h_{(p)} (\vct{b})&\coloneqq\sum_{m=1}^M h_m^p (\vct{b})\nonumber
\end{align}
The following rigorous mathematical result justifies the approximation symbol ''$\approx$'' in this display. 

\begin{proposition} \label{prop:approx}
For each $p \in \mathbb{N}_+$, the
 following sandwich inequality is true for all binary vectors $\vct{b} \in \left\{ 0, 1\right\}^n$:
\begin{align}
M^{-1/p} \left( h_{(p)} (\vct{b})\right)^{1/p} \leq h_{\max}(\vct{b}) \leq \left( h_{(p)} (\vct{b})\right)^{1/p}.\nonumber
\end{align}
\end{proposition}

\begin{proof}
Recall that for $p \in \mathbb{N}_+$, the $\ell_p$-norm of a vector $\vct{u} \in \mathbb{R}^M$ is defined as $\| \vct{u} \|_p = \big(\sum_{m=1}^M |u|_m^p\big)^{1/p}$.
For any fixed binary input vector $\vct{b}$, we define the $M$-dimensional vector $\vct{u}$ entrywise by setting $u_m = h_m (\vct{b})$. Our non-negativity assumption ensures that we can identify the following $\ell_p$-norm reformulations of the objects of interest: $h_{\max}(\vct{b})
= \max \left\{ u_1,\ldots,u_M \right\}=\|\vct{u} \|_\infty$, as well as $h_{(p)} (\vct{b})
= \sum_{m=1}^M u_m ^{p} =\left\| \vct{u} \right\|_{p}^{p}$. The claim now follows from a well-known $\ell_p$-norm relation in $\mathbb{R}^M$, see e.g.~\cite[Eq.~(A.3)]{Foucart_2013}:
$
M^{-1/p} \left\| \vct{u} \right\|_p \leq \left\| \vct{u} \right\|_\infty \leq \left\| \vct{u} \right\|_p.\nonumber
$
\end{proof}

We emphasize that this sandwich inequality simultaneously ranges over all $2^n$ different inputs $\vct{b} \in \left\{0,1\right\}^n$, while the approximation accuracy is completely independent of $n$. In fact, choosing $p$ proportional to $\log_2 (M)$ ensures that upper and lower bounds 
match up to a multiplicative constant that can be made arbitrarily small.~%
To reliably recover the optimum of $h_{\max}$ via the optimum of the approximation $h_{(p)}$, the difference between these bounds needs to be sufficiently small.
This is formulated via the spectral gap ratio as a measure of the separation between the optimum and the first non-optimal assignment.

\begin{reptheorem}{thm:theorem_1}
Let $h_{\max}(\vct{b})$ be an objective function defined as the maximum of $M$ objectives $h_m(\vct{b)}$ with minimal argument $\vct{b^*}$ and minimal value $\lambda_1:=h_{\max}(\vct{b^*})$.
    Choosing an approximation level 
    \begin{align}\label{eq:required_p}
        p > p_0 = \frac{\log(M)}{\log(r(h_{\max}) + 1)}
    \end{align}
    ensures that binary vector $\vct{b^*_{(p)}}$ that minimizes $h_{(p)}$ also minimizes $h_{\max}$.
    The parameter $r(h_{\max})$ is the \emph{spectral gap ratio} $(\lambda_2 - \lambda_1) / \lambda_1$, where $\lambda_2$ is the objective function value of the first non-optimal solution.
\end{reptheorem}

\begin{proof}
For the proof, we first assume that there is exactly one binary vector $\vct{b}_g$ (the subscript $g$ stands for ''ground state'') that minimizes $h_{\max}(\vct{b})$ and show that for $p > p_0$, this vector also minimizes $h_{(p)}$.

We use the l.h.s. of \Cref{prop:approx} to bound the objective value of the $p$-th approximation at $\vct{b}_g$ as
\begin{align}\label{eq:lower}
    h_{(p)}(\vct{b}_g)\leq M(h_{\max}(\vct{b}_g))^p.
\end{align}
Similarly, we use the r.h.s. of \Cref{prop:approx} to bound the objective values of these two functions at any other argument $\vct{b}_e \neq \vct{b}_g$ (the subscript ''$e$'' stands for excited state) as
\begin{align}\label{eq:upper}
    h_{(p)}(\vct{b}_e)\geq (h_{\max}(\vct{b}_e))^p.
\end{align}
Now, note that the approximation $h_{(p)}$ has the same (unique) optimal solution $\vct{b}_g$ as $h_{\max}$ if and only if $h_{(p)}(\vct{b}_g) < h_{(p)}(\vct{b}_e)$ for all $\vct{b}_e \neq \vct{b}_g$, or equivalently:
\begin{align*}
\min_{\vct{b}_e \neq \vct{b}_g} \frac{h_{(p)}(\vct{b}_e)}{h_{(p)}(\vct{b}_g)} >1.
\end{align*}
We can now apply \cref{eq:lower} to the denominator and \cref{eq:upper} to the numerator of this ratio to obtain a lower bound that only involves the original objective function:
\begin{align}
\min_{\vct{b}_e \neq \vct{b}_g} \frac{h_{(p)}(\vct{b}_e)}{h_{(p)}(\vct{b}_g)}
 \geq \min_{\vct{b}_e \neq \vct{b}_g}\frac{1}{M} \left( \frac{h_{\max}(\vct{b}_e)}{h_{\max}(\vct{b}_g)}\right)^p. \label{eq:aux-ineq}
\end{align}
To make a connection to the spectral gap ratio of the objective $h_{\max}(\vct{b})$, we observe $h_{\max}(\vct{b}_g)=\lambda_1$ (ground state energy) and $h_{\max}(\vct{b}_e) \geq \lambda_2$ (first excited state energy) for all $\vct{b}_e \neq \vct{b}_g$. In formulas:
\begin{align*}
\min_{\vct{b}_e \neq \vct{b}_g}\frac{1}{M} \left( \frac{h_{\max}(\vct{b}_e)}{h_{\max}(\vct{b}_g)}\right)^p =
\frac{1}{M} \left( \frac{\lambda_2}{\lambda_1}\right)^p
= \frac{ \left(r (h_{\max})+1\right)^p}{M},
\end{align*}
where $r(h_{\max})=(\lambda_2-\lambda_1)/\lambda_1$ denotes the spectral gap ratio. This reformulation tells us that choosing
\begin{equation}
p > p_0= \frac{\log(M)}{ \log \left(r(h_{\max})+1\right)}
\label{eq:aux-p}
\end{equation}
is enough to ensure
\begin{align*}
\frac{1}{M} \left( r(h_{\max})+1\right)^p >&
\frac{1}{M}\left(r(h_{\max})+1\right)^{\frac{\log (M)}{\log( r(h_{\max})+1)}} \\
=& \frac{M}{M}=1.
\end{align*}
And, courtesy of \cref{eq:aux-ineq}, this then also ensures that $\min_{\vct{b}_e \neq \vct{b}_g} h_{(p)}(\vct{b}_e)/h_{(p)}(\vct{b}_g)>1$.
And, as explained at the beginning of the proof, this is equivalent to demanding that the approximating function $h_{(p)}(\vct{b})$ has exactly the same (unique) minimum as $h_{\max}(\vct{b})$.

The essential insights from the proof that we just presented don't change if we move to a degenerate objective function $h_{\max}(\vct{b})$ with $L \geq 2$ optimal solutions $\vct{b}_{g}^{(1)},\ldots,\vct{b}_{g}^{(L)} \in \left\{0,1\right\}^n$.
Let $\vct{b}_e$ be any assignment with $h_{\max}(\vct{b}_e)$ not minimal, then using the reasoning above,
\begin{align}
    \frac{h_{(p)}(\vct{b}_e)}{h_{(p)}\left(\vct{b}_{g}^{(l)}\right)} \geq \frac{1}{M} \left( \frac{\lambda_2}{\lambda_1}\right)^p > \frac{ \left(r (h_{\max})+1\right)^{p_0}}{M} = 1,
\end{align}
for all $\vct{b}_{g}^{(l)}$.
Therefore, for $p > p_0$, $h_{(p)}(\vct{b}_e) \geq h_{(p)}\left(\vct{b}_{g}^{(l)}\right)$ is guaranteed.
Conversely, the optimum $\vct{b_{(p)}^*}$ of the approximation $h_{(p)}$ has to be a ground state $\vct{b}_{g}^{(l)}$.
\end{proof}

\begin{figure}
\centering
\begin{tikzpicture}[x=1cm,y=1cm,>=stealth]
\def\gnd{0}
\def\exc{1.6}
\def\gndA{0.00}
\def\gndB{0.25}
\def\excB{2.40}

\begin{scope}[shift={(0,0)}]
  \draw[gray!35, line cap=round] (-1.3,\gnd) -- (1.3,\gnd);
  \draw[gray!35, line cap=round] (-1.3,\exc) -- (1.3,\exc);
  \draw[line width=1.1pt] (-1.10,\gnd) -- (-0.50,\gnd);
  \draw[line width=1.1pt] (-0.30,\gnd) -- ( 0.30,\gnd);
  \draw[line width=1.1pt] ( 0.80,\exc) -- ( 1.20,\exc);
  \draw[<->, very thick] (1.15,\gnd) -- node[right=2pt, text=red] {$r(h_{\max})$} (1.15,\exc);
  \node[below] at (-0.15,-0.45) {$h_{\max}$};
\end{scope}

\begin{scope}[shift={(5,0)}]
  \draw[gray!35, line cap=round] (-1.6,\gndA) -- (1.6,\gndA);
  \draw[gray!35, line cap=round] (-1.6,\gndB) -- (1.6,\gndB);
  \draw[gray!35, line cap=round] (-1.6,\excB) -- (1.6,\excB);
  \draw[line width=1.1pt] (-1.20,\gndA) -- (-0.60,\gndA);
  \draw[line width=1.1pt] (-0.10,\gndB) -- ( 0.50,\gndB);
  \draw[line width=1.1pt] ( 1.00,\excB) -- ( 1.45,\excB);
  \draw[<->, thick] (-1.35,\gndA) -- node[left=2pt, text=red] {$r_{(p)}$} (-1.35,\gndB);
  \draw[<->, very thick] (1.35,\gndB) -- node[right=2pt, text=red] {$\tilde{r}_{(p)}$} (1.35,\excB);
  \node[below] at (-0.05,-0.45) {$h_{(p)}$};
\end{scope}

\end{tikzpicture}

\caption{
    On the left, we show a schematic of the solution space of some $h_{\max}$ with a degenerate minimum. 
    Measuring the objective value ($y$-axis) in units of $\lambda_1$, the gap ratio $r(h_{\max})$ can be depicted as the vertical shift between the degenerate optimal solutions and the best non-optimal solution.
    On the right, we show the schematic solution space of $h_{(p)}$.
    The approximation will in general break that symmetry, leaving only one global minimum.
    Although the spectral gap $r_{(p)}$ on the right-hand side is smaller, the separation $\tilde{r}_{(p)}$ between any minimum and the first non-optimal state of $h_{\max}$ increases.}
\label{fig:groundspace}
\end{figure}

The arguments above show that, for high enough $p$, the $p$-th approximation separates the ground state space of $h_{\max}$ from all non-optimal solutions. 
However, it can, and typically will, happen that the degeneracy among the different ground states is broken. 
So, the optimal solution of the approximation $h_{(p)}$ is promised to be one optimal solution of the original problem $h_{\max}$.
We visualize this broken degeneracy in \Cref{fig:groundspace}.

The spectral gap is often used to quantify the required resources to solve an optimization problem.
With that in mind, the effect of the approximation procedure on the spectral gap is evaluated.
Consider again the case of a unique ground state $\vct{b}_g$.
We start from the definition of the spectral gap ratio, where we denote the minimal eigenvalue of $h_{(p)}$ with $\nu_{1}$ and the second lowest energy with $\nu_{2}$:
\begin{align}\label{eq:ratio_p}
    r(h_{(p)})
    =\frac{\nu_2-\nu_{1}}{\nu_{1}}
    =\frac{\nu_2}{\nu_{1}}-1.
\end{align}
We can now reuse the correspondence between the objective function and eigenvalues to bound $\nu_{1}\leq M\lambda_{1}^p$ using \cref{eq:lower}. In the same way, we bound the energy of the first excited state as $\nu_{2}\geq\lambda_{2}^p$ with \cref{eq:upper}.
Inserting those bounds into the denominator and numerator of \cref{eq:ratio_p}, respectively, yields:
\begin{align*}
    r(h_{(p)})&=\frac{\nu_2}{\nu_{1}}-1
    \geq\frac{\lambda_2^p}{M\lambda_{1}^p}-1.
\end{align*}
This enables us to rewrite 
\begin{align*}
    r(h_{(p)})\geq r(h_{\max})\Leftrightarrow\frac{1}{M}\left(\frac{\lambda_2}{\lambda_{1}}\right)^p-1\geq\frac{\lambda_2}{\lambda_{1}}-1.
\end{align*}
The r.h.s. is fulfilled whenever $\frac{1}{M}\left(\frac{\lambda_2}{\lambda_{1}}\right)^{p-1}\geq1$ or equivalently $p\geq1+\frac{\log(M)}{ \log \left(r(h_{\max})+1\right)}$. 
Thus, by using $p_0+1$ (with $p_0$ from \Cref{thm:theorem_1}), we can even ensure that, in addition to the preservation of the global minimum (\cref{eq:aux-p}), the spectral gap ratio is increased.

The previous argument derives an increase in the spectral gap ratio for $h_{\max}$ having only one unique optimum.
The shown increase in the spectral gap ratio for $p>p_0+1$ is based on the energies of the minimal states rather than their arguments.
Therefore, by the arguments above, also in the degenerate case, all vectors that are minima of $h_{\max}$ will be relatively separated from any non-ideal solution by more than $r(h_{\max})$ in $h_{(p)}$.
Thus, while the difference between the optimum and the first non-optimal solution might be smaller in $h_{(p)}$ in the degenerate case (see \Cref{fig:groundspace}), the latter solution will likely correspond to another state that is optimal for $h_{\max}$.
Consequently, finding that solution would have no negative impact. 

We have chosen to phrase
\Cref{thm:theorem_1} as a threshold theorem: based on a given number of objectives and a given spectral gap ratio, we can determine a finite $p$ that guarantees that the \emph{MOQA} approximation becomes exact. The following corollary of this result offers a different perspective on the same insight, but replaces knowledge about the spectral gap with accepting a small pre-specified error.

\begin{repcorollary}{corr:approx}
Let $\vct{b^*}$ be the binary vector that minimizes $h_{\max}(\vct{b})$, an objective function defined as the maximum of $M$ objectives $h_m(\vct{b)}$.
For a desired approximation ratio $c \geq 1$, choosing an approximation level $p > \frac{\log(M)}{\log(c)}$ ensures that the vector $\vct{b^*_{(p)}}$ minimizing $h_{(p)}$ obtains an approximation ratio of at most $c$.
\end{repcorollary}
\begin{proof}
For a proof by contradiction, we assume an approximation ratio $h_{\max}(\vct{b^*_{(p)}})/h_{\max}(\vct{b^*}) > c$. 
By applying \Cref{prop:approx} to bound the numerator from above and the denominator from below, we infer
\begin{align}
     &&\left(\frac{Mh_{(p)}\left(\vct{b^*_{(p)}}\right)}{h_{(p)}(\vct{b^*)}}\right)^{1/p} &> c\\
     \Leftrightarrow&& \frac{h_{(p)}\left(\vct{b^*_{(p)}}\right)}{h_{(p)}(\vct{b^*)}} &> \frac{c^p}{M}.
\end{align}
Substituting $p$ with the required approximation level then allows to lower bound the fraction as
\begin{align}
    \frac{h_{(p)}\left(\vct{b^*_{(p)}}\right)}{h_{(p)}(\vct{b^*)}} > \frac{c^{\frac{\log(M)}{\log(c)}}}{M} = 1,
\end{align}
contradicting the assumption that $\vct{b^*_{(p)}}$ minimizes $h_{(p)}$.
\end{proof}

\subsection{Ising Form and Additive Shift}

As preparation for the derivation of the classical pre-processing algorithm required when using \emph{MOQA} for quantum optimization, we define the representation of diagonal Hamiltonians via spin variables in this section. 
We exhibit this process explicitly for the case of degree $\leq 2$ objectives.
Since \emph{MOQA} requires all objectives to be non-negative, we further describe a method of calculating valid additive shifts from the spin-variable formulation for QUBOs.

Starting from \cref{eq:quadratic-function}, the corresponding operator $\hat{H}$ can be derived by exchanging the binary variables $\vct{b} \in \left\{0, 1\right\}^n$ with the spin-variables $\vct{s} \in \left\{-1,+1\right\}^n$ while keeping the same objective values via the relationship $\vct{b} = \frac{1}{2}(\vct{1} - \vct{s})$.
For the quadratic objectives, this leads to the function
\begin{align}
    h(\vct{b})&=\tilde{h}(\vct{s})=\vct{s}^\transp \mtx{A} \vct{s} + \vct{a}^\transp \vct{s}+\alpha\label{eq:ising_compact}\\
    &=\sum_{i,j}^n \left[\mtx{A}\right]_{i,j} s_i s_j + \sum_{i=1}^n \left[\vct{a}\right]_i s_i + \alpha,\label{eq:ising}
\end{align} 
with $\mtx{A}=\mtx{M}/4$, ${\vct{a} = -\vct{1}^\transp \mtx{M}/2}$ and $\alpha=\vct{1} ^\transp\mtx{M} \vct{1}/4$.
In the last step, use $s_i=\langle b_i|Z_i|b_i\rangle$ to replace $s_i$ with Pauli-$Z_i$ observables to arrive at
\begin{align}
\hat{H} &= \sum_{i,j}^n \left[\mtx{A}\right]_{i,j} Z_i Z_j + \sum_{i=1}^n \left[\vct{a}\right]_i Z_i + \alpha I. \label{eq:Hamiltonian}
\end{align}

In comparison, the Ising formulation (\cref{eq:ising}) generally has more variables than its binary counterpart, but still cannot describe larger sets of problems. The reason is that in the original formulation, the linear part is equivalent to the main diagonal of $\mtx{M}$ because of $b_i^2=b_i$. Those $n$ variables are separated in the Ising formulation to $\vct{a}$, because $s_i^2=1$ makes the additional variables on the main diagonal of $\mtx{A}$ equal to a constant shift.

The reduction formalized in \Cref{thm:theorem_1} is possible for non-negative objective functions, which we ensure for general problems via additive shifts. 
The goal is to have a minimal shift $C$ that still guarantees a positive landscape, that is, ${h_m(\vct{b})+C\geq0\;\forall\;\vct{b}\in\{0,1\}^n,m\in\{1,\dots,M\}}$. 
Depending on the problem at hand, common strategies for obtaining lower bounds can be employed.
For example, for linear binary objectives, a simple relaxation to linear real-valued objectives suffices, while for quadratic problems a common strategy is to employ semidefinite programming relaxations.
In the following, we present a derivation of a shift-value based on the minimal eigenvalues of the quadratic problems.

Starting from \cref{eq:ising_compact}, the linear and quadratic terms can be merged by extending the underlying solution space to $\tilde{\vct{s}}\in\{\pm1\}^{n+1}$:
\begin{align}
    \tilde{h}(\vct{s})
    &=
    \left(
    \begin{array}{c}
     \\
     \vct{s}\\
    \\
    \hline
    1
    \end{array}
    \right)^\transp
    \underbrace{
    \left(
    \begin{array}{c|c}
    \begin{array}{ccc}
     & & \\
     & \mtx{A} & \\
     & &
    \end{array}
    &
    \vct{a}/2 \\ \hline
    \vct{a}^{\transp}/2 & 0
    \end{array}
    \right)}_{\tilde{A}}
    \left(
    \begin{array}{c}
     \\
     \vct{s}\\
    \\
    \hline
    1
    \end{array}
    \right)+\alpha.
\end{align}
The quadratic form in the first summand is bounded below via its (scaled) smallest eigenvalue ${\tilde{\vct{s}}^T \tilde{\mtx{A}}\vct{s}\geq\lambda_{\min}(\tilde{\mtx{A}})\lVert \tilde{\vct{s}}\rVert^2}$, where $\lVert\tilde{\vct{s}}\rVert^2 = n+1$ by definition.
By iterating over all $m$ objectives, we therefore obtain the shift ${C=\underset{m\in\{1\dots M\}}{\max}\left\{-(n+1)\lambda_{\min}(\tilde{\mtx{A}}_m)-\alpha_m\right\}}$.
The above approach reduces the problem to the minimization of a quadratic form and could be extended to higher degrees, e.g., by reducing degree-$k$ polynomials to cubic forms.
However, the computational resources required to bound its eigenvalues could hinder such a generalization, and other approaches~\cite{Lasserre_2001} might be more appropriate.

\subsection*{Pre-Processing and Circuit Complexity}

For the implementation of the pre-processing logic, we assume that each of the $M$ objectives ${\hat{H}_m = \sum_{j=1}^T c_{mj} P_{mj}}$ is represented as a data structure mapping Pauli strings $P_{mj} \in \{I,Z\}^{\otimes n}$ to their coefficients $c_{mj} \in \mathbb{R}$.
Such a representation is standard in common quantum computing libraries when encoding optimization problems~\cite{pennylane_paulistring, Krotz_2026, Dion_2024}.
However, we further require that iterating over all Pauli strings is possible in time $\mathcal{O}(T)$ and that retrieval and storage of coefficients is done in time $\mathcal{O}(n)$ (while still taking $\mathcal{O}(T)$ memory).
The second requirement in particular depends on the implementation of the underlying data structure.
For example, implementations based on hash maps (e.g.~\cite{pennylane_paulistring}) could, although fast on average, lead to a worst-case access time in $\Theta(T)$.
Fortunately, since for degree-$k$ polynomials, the Hamming weight of each $P_{mj}$ is bounded, simplifications are possible. 
One data structure that fulfills all requirements above is a list of coefficients $c_{mj}$ ordered in lexicographic order of the Pauli-Z indices in $P_{mj}$.
For strings of maximal Hamming weight $k \leq n$, for example, the combinatorial number system can be used for obtaining the index of a particular $P_{mj}$ using $\mathcal{O}(k)$ table lookups~\cite{Knuth_2022, Wallerberger_2022}.

The implementation of the classical pre-processing required to obtain $\hat{H}_{(p)}$ then in practice simply amounts to addition and multiplication of these objects.
In the following, we derive the time complexity of this process.

\begin{repproposition}{prop:runtime}
Given a list of $M$ degree-$k$ Hamiltonians in $n$ variables and an approximation level $p$, the degree-$kp$ polynomial $\hat{H}_{(p)}$ approximating $\hat{H}_{\max}$ is computed in time $\mathcal{O}\left(Mn^{kp+1}p\right)$ and contains $\mathcal{O}\left(n^{kp}\right)$ terms. 
\end{repproposition}
\begin{proof}
First, note that any Pauli string in a classical objective can be described by the subset of qubit indices to which Pauli-$Z$ is applied.
As there are $\sum_{j=1}^k\binom{n}{j}$ such subsets for degree $\leq k$ objectives, we can bound the number of terms $T$ in each Hamiltonian as ${T \leq \sum_{j=1}^k\binom{n}{j} \in \mathcal{O}(n^k)}$.
Raising any degree-$k$ polynomial to the $p$-th power produces (at most) degree-$kp$ polynomials.
Further, adding all $M$ resulting polynomials does not increase the degree; thus the resulting operator contains $\mathcal{O}(n^{kp})$ terms.

The multiplication step itself then proceeds akin to standard term-by-term expansion:
To raise a Hamiltonian to the $p$-th power, we calculate the product of all $p$-tuples composed of the individual terms in $\hat{H}_m$.
Each of the $p$ individual factors is a string of $n$ elements from $\{I,Z\}$.
Because of involutority, their multiplication amounts to the addition modulo $2$ of $n$-bit binary strings.
In total, these multiplications are therefore performed in $\mathcal{O}(np)$, and the resulting coefficient is updated in the data structure representing $\hat{H}_{(p)}$.
Since $\hat{H}_{(p)}$ has maximal degree $kp$ the insertion/update of the coefficent requires $\mathcal{O}(kp)$ time (or $\mathcal{O}(n)$ if $n\geq kp$).
Consequently, this step is entirely dominated by the multiplication time.

Since there are $T^p$ tuples to consider for an operator containing $T$ terms, we arrive at a classical pre-processing time of $O(T^p pn)$ for one operator.
As a result, for $M$ operators of $T\in\mathcal{O}(n^k)$ terms, $\hat{H}_{(p)}$ is obtained in time $O(Mn^{kp + 1}p)$.
\end{proof}

Since the degree and maximal number of terms in $\hat{H}_{(p)}$ are known in advance, the algorithm that is sketched above also avoids reserving memory for all $2^n$ Pauli strings, keeping space complexity in $\mathcal{O}(n^{kp})$.
For the highlighted case of quadratic optimization ($k=2$) in the main text, we arrive at an overall time complexity of $\mathcal{O}(Mn^{2p+1}p)$ (\Cref{prop:runtime}).
If both $p$ and $M$ are constant (independent of $n$) for the problem at hand, then the runtime is polynomial in problem size.
However, depending on $p$, the preprocessing time can still require substantial classical resources.
Thus, using libraries specifically tailored to streamlining Pauli operator arithmetic~\cite{Dion_2024, Krotz_2026} can be beneficial for this task.

From the perspective of the implemented quantum circuits, \Cref{prop:runtime} bounds the number of abstract degree-$kp$ operators before transpilation, which can be used as a rough estimate of the final circuit size. 
Additionally, it also bounds the number of different quantum couplings one needs to realize to minimize the \emph{MOQA} Hamiltonian on a quantum device. This is a meaningful summary parameter for the \emph{quantum cost} associated with running \emph{MOQA}.

\subsection{Study Details}

In our numerical studies in the main text, we rely on two error metrics evaluated over $N_s$ randomly generated problem instances. With the notation ${\vct{b}^{*}=\text{argmin}(h_{max}(\vct{b}))}$ from before, we define a 0-1 type error capturing the (average) success probability in recovering the optimum:
\begin{equation*}
\epsilon= \frac{1}{N_s}\sum_{i=1}^{N_s}\begin{cases}
        \text{0 if }h_{\max}^{(i)}(\vct{b}_{(p)}^{*(i)}) \neq h_{\max}^{(i)}(\vct{b}^{*(i)})\\
        \text{1 else }.
\end{cases}
\end{equation*}
Here, $h_{\max}^{(i)}(\vct{b}_{(p)}^{*(i)})$ is the value of the desired objective function $h_{\max}^{(i)}(\vct{b})$ evaluated at the minimal argument $\vct{b}_{(p)}^{*(i)}$ of $h_{(p)}^{(i)}(\vct{b})$ for the $i$-th random problem instance. This metric serves to evaluate \emph{MOQA} against exact classical methods. To make a similar comparison against classical heuristics, we evaluate the (average) approximation ratio of the minima given by:
\begin{equation*}
    \delta = \frac{1}{N_s}\sum_{i=1}^{N_s}\frac{h_{\max}^{(i)}(\vct{b}_{(p)}^{*(i)})}{h_{\max}^{(i)}(\vct{b}^{*(i)})}.
\end{equation*}
We numerically evaluate these metrics by finding $\vct{b^*}$ and $\vct{b_{(p)}^*}$ using brute-force search.
The implementations used to obtain our results, the code for their reproduction, and additional functions for generating the quantum operators for \emph{MOQA} are available online~\cite{MOQA_repo}.

For our first numerical evaluation (\Cref{fig:ineq_constraints_b}), we generate random linearly constrained quadratic problems.
All variables are sampled from a standard Gaussian distribution (zero mean and unit variance).
For the second evaluation (\Cref{fig:nae_approx_error}), we randomly generate NAE-3SAT formulas of $n=20$ variables. 
For each evaluated clause-to-variable ratio $r$, we set $M = rn$.
Here, $M$ refers to the number of clauses in the formula, so each formula defines an optimization problem with $2M$ constraints. 
A formula is specified via literals described by the indices $m$ (for the clause) and $i$ (position within the clause). 
For each combination $m,i$ we assign one atom $b_{l{m,i}}$, where $0\leq l_{m,i} <n $ is the variable index of the atom. 
In addition to the atom, each $m,i$ is assigned a sign $\sigma_{m,i} \in \{\pm 1\}$ that encodes its polarity, i.e., a negated atom $\lnot b_{l{m,i}}$ is $-b_{l{m,i}}$.
Using spin variables $s_{l_{m,i}} = 1-2b_{l_{m,i}}$, we have 
\begin{align}
    \sigma_{m,i} s_{l_{m,i}} = \begin{cases}
    -1 & \text{if literal is satisfied}\\
    1 & \text{if literal is not satisfied}
    \end{cases}.
\end{align}
Summing over all three of these products corresponding to the three literals in the $m$-th clause therefore amounts to the difference between the number of unsatisfied and satisfied literals in the clause:
\begin{align}
    \sum_{i=1}^3 \sigma_{m,i} s_{l_{m,i}} = \mathrm{Numunsat}_m(\vct{b}) - \mathrm{Numsat}_m(\vct{b})
\end{align}
Since, $\mathrm{Numunsat}_m(\vct{b}) = 3-\mathrm{Numsat}_m(\vct{b})$, we obtain a linear function 
\begin{align}
    \mathrm{Numsat}_m(\vct{b}) = \frac{1}{2}\left( 3- \sum_{i=1}^3 \sigma_{m,i} s_{l_{m,i}}\right)
\end{align}
to be used for encoding the constraints
\begin{align}\label{eq:constraints_as_numsat}
1 \leq \mathrm{Numsat}_m(\vct{b}) \leq 2 &\qquad \text{for $1\leq m \leq M$}.
\end{align}

To encode such optimization problems containing multiple constraints, we use a generalization of the process sketched in the introduction of this work:
\begin{proposition} \label{prop:ineq}
Every optimization problem with $M$ inequality constraints
\begin{align*}
\underset{\vct{x} \in S}{\text{minimize}} \quad & h(\vct{x}) \\
\text{subject to} \quad & g_m (\vct{x}) \geq 0 \quad \text{for $m=1,\ldots,M$}
\end{align*}
can be regularized to produce an equivalent problem that is the maximum of $(M+1)$ unconstrained objective functions. 
\end{proposition}
\begin{proof}
First, each individual constraint is rewritten as:
\begin{equation*}
    g_m(\vct{x}) \geq 0 \leftrightarrow \max\{0,-g_m(\vct{x})\} = 0.
\end{equation*}
These are now equality constraints that can be added to the original objective. If multiplied by large enough penalties $\gamma_m$, these additives ensure that the minimum of the new function does not violate a constraint. For simplicity, we set $\gamma_m=\gamma$ and further rewrite:
\begin{align*}
    h_{\max}(\vct{x})&=h(\vct{x}) +\sum_{m=1}^M \gamma\max\{0,-g_m(\vct{x})\}\\
    &=h(\vct{x}) +\gamma \underset{\vct{v}\in\{0,1\}^M}{\max}\left\{-\sum_{m=1}^M v_mg_m(\vct{x})\right\}\\
    &\approx h(\vct{x}) +\gamma \underset{\substack{\vct{v}\in\{0,1\}^M\\ ||\vct{v}||_1\leq 1}}{\max}\left\{-\sum_{m=1}^M v_mg_m(\vct{x})\right\}\\
    &= h(\vct{x}) +\gamma\max\{0,\underset{m\in\{1\dots M\}}{\max}\left\{ -g_m(\vct{x})\right\}\}\\
    &=\max\{h(\vct{x}),h(\vct{x})-\gamma g_1(\vct{x}),\dots,h(\vct{x})-\gamma g_M(\vct{x})\}.
\end{align*}
In the first to the second line, the sum over $M$ maxima of two terms is converted into a maximum over $2^M$ terms, considering all combinations. However, in the context of minimizing over $\vct{x}$, we can simplify from the second to the third line. The $\approx$ is justified if, for large enough $\gamma\gg0$, violating one constraint ($||\vct{v}||_1= 1$) is already enough to shift the cost sufficiently far away. Thus, all terms that would penalize $\vct{x}$ even more if they violate multiple constraints are neglected. This can then be rephrased as a single maximum over only $M+1$ elements.
\end{proof}

While \Cref{prop:ineq} is applied here specifically to address the NAE-3SAT problem, we emphasize that this reduction applies to all inequality-constrained optimization problems.
Furthermore, the number of objective functions only grows linearly in the number of inequality constraints. This is much better than the naive scaling of $2^M$ one would obtain by merely iterating the regularization of individual inequality constraints one at a time.

After the application of \Cref{prop:ineq} to the NAE-3SAT problem, we arrive at a maximum over $2M+1$ objective functions.
In addition to the objective function of the original problem, we also have functions derived from 
\begin{align}
    g_{m,\text{lower}}(\vct{b}) := \mathrm{Numsat}_m(\vct{b}) - 1 \geq 0
\end{align}
for each lower bound constraint in \cref{eq:constraints_as_numsat}, and 
\begin{align}
    g_{m,\text{upper}}(\vct{b}) := 2-\mathrm{Numsat}_m(\vct{b}) \geq 0
\end{align}
for each upper bound.
The maximization is then approximated using \emph{MOQA} for $p\in\{2,\dots,8\}$ to evaluate the success probabilities presented in \Cref{fig:nae_approx_error}.
To this end, we select $\gamma=n+1$ as the regularization parameter.
Since violation of any constraint implies that there is some $m$ with $g_{m,\text{lower}}(\vct{b}) \leq -1$ or $g_{m,\text{upper}}(\vct{b}) \leq -1$.
This ensures that $h_{\max}(\vct{b}) \geq n+1$ (when derived as in \Cref{prop:ineq}), which is larger than the objective function value of any feasible solution. This derivation uses that the objective function $\sum_i^n b_i$ is bounded to $[0,n]$.
By bounding the Hamming weight of satisfying assignments for the input formula, $\gamma$ could be reduced further.

The constrained optimization version of the NAE-3SAT problem has no feasible solutions if the formula itself is unsatisfiable.
In this case, we have no ground truth in our evaluation to compare the \emph{MOQA} approximation to. 
Consequently, we post-selected for satisfiable instances in our experiments, while still making sure we have $420$ instances for each evaluated data point in \Cref{fig:nae_approx_error}.
Our experiments, however, show that \emph{MOQA} can reliably distinguish between unsatisfiable and satisfiable instances of the problem.
A simple (classical) verification of whether the returned assignment $\vct{b}^*_{(p)}$ of the approximation $h_{(p)}$ satisfies the formula can be used as a reliable predictor for satisfiability of the formula as a whole.
The additional evaluation provided as supplementary material shows that this approach has success probability $\geq 80\%$ for $p=6$ and even $\geq 95\%$ for $p=8$ consistently for all evaluated clause-to-variable ratios.

\section*{Code availablility}
The code used to produce our numerical experiments, the raw experiment results and the code used to generate the visualizations is available at \url{https://github.com/SebastianEgginger/MOQA}.

\section*{Acknowledgments}
We want to thank Raimel Medina Ramos and Frank Leymann for inspiring discussions. Johannes Kofler and Sergi Ramos Calderer provided helpful feedback on the paper draft. This research was funded by the Austrian Research Promotion Agency (FFG) within the COMET module Quantum Algorithm Engineering (FFG grant no. 923923) and via the QuantumReady project (FFG 896217) within Quantum Austria. The QUICK-team is also supported by the Austrian Science Fund (FWF) via the SFB BeyondC (10.55776/FG7) and the European Research Council (ERC) via the Starting grant q-shadows (101117138).

\end{document}